\documentclass[11pt,english]{article}
\usepackage[T1]{fontenc}
\usepackage[latin9]{inputenc}
\usepackage{listings}
\usepackage{geometry}
\usepackage{amssymb}
\usepackage{graphicx}
\usepackage{babel}
\usepackage{tikz}
\usepackage{amsmath,amsthm}
\usepackage{algorithm}
\usepackage{mathtools}
\usepackage{amsmath}
\newcommand{\norm}[1]{\left\lVert#1\right\rVert}
\newtheorem{theorem}{Theorem}
\newtheorem{example}{Example}

\newcommand{\modop}{\,\, {\rm mod }\,\,}

\graphicspath{ {images/} }
\begin{document}

\title{A practical attack to Bouftass's cryptosystem}

\author{Yang Zhang}

\maketitle
\bigskip{}

{

\centering Texas Tech University

\centering Lubbock, TX, United States

\centering yang22.zhang@ttu.edu

}

\bigskip{}

\bigskip{}

\begin{abstract}
Recently, a new fast public key exchange protocol was presented by S. Bouftass.
The protocol is based on the difficulty of inverting the function
$F(x)=\lfloor (zx \modop 2^p)/ 2^q \rfloor$. In this paper, we describe a
practical attack against this protocol based on Closest Vector Problem (CVP) and Gaussian lattice reduction.
\end{abstract}
\smallskip
\noindent \textbf{Keywords}\quad public key exchange\quad cryptoanalysis\quad CVP\quad Gaussian lattice reduction

\section{Introduction}

In public key cryptography, the security of traditional methods is based on number theoretic problems, and suffers from high computational cost due to problems such as dealing with large numbers. Each user in a public key system has a pair of cryptographic keys, consisting a public key and a private key. These are related through a hard mathematical inversion problem, so that the private key cannot be feasibly derived from the public key.   A standard implementation of public key cryptography is based on the Diffie-Hellman key agreement protocol\cite{key-1}. This protocol allows two users, Alice and Bob, to exchange a secret key over an insecure communication channel. It can be described as following:
\begin{enumerate}
\item Alice and Bob openly agree upon a large prime $p$ and $g\in\mathbb{Z}_{p}^*$.
\item Alice randomly chooses the secret integer $a\in [1, p-1]$.
\item Alice computes $A=g^{a}\mod p$, and publishes $A$.
\item Bob randomly chooses the secret integer $b\in [1, p-1]$,
\item Bob computes $B=g^{b}\mod p$, and publishes $B$.
\item Alice computes the secret integer $K_{A}=B^{a}\mod p=g^{ba}\mod p$.
\item Bob computes the secret integer $K_{B}=A^{b}\mod p=g^{ab}\mod p$.
\end{enumerate}

Then Alice and Bob can get the same shared secret key $K=K_{A}=K_{B}$. The eavesdropper Eve knows $p,g,A$ and $B$, and she needs
to compute the secret key $K$. For this, it suffices to solve one of the discrete logarithm
problems:
\begin{eqnarray*}
A=g^{a}\mod p \qquad\text{and}\qquad B=g^{b}\mod p
\end{eqnarray*}
for the unknowns $a$ or $b$. If $p$ is a very large prime of say 2048 bits, then the problem
becomes computationally hard, and
it is considered infeasible. For maximum security $p$ should be a safe prime,
i.e. $(p-1)/2$ is also a prime, and $g$ a primitive root of $p$ \cite{key-2}.

 Recently, to construct a cryptosystem which is not based on number theory, S.Bouftass described a new public key exchange protocol relying on the difficulty of inverting the function $F(x)=\lfloor (zx \modop 2^p)/ 2^q \rfloor$ \cite{key-3}.  In our work, we find that this system is not secure, we can easily break this system based on the closest vector problem (\cite{key-4},\cite{key-5}) and Gaussian lattice reduction\cite{key-4}. This paper is organized as follows, in section 2 we give a general description of S.Bouftass's new protocol; section 3 gives our method to break this system and an example; The last section is conclusion.

\section{S.Bouftass's new public key exchange cryptosystem}

Throughout, if $n$ is an integer and $s\in \mathbb{N}$, we use $a\mod n$ to denote the nonnegative reminder of $a$ divided by $n$. We will use the same notation as in \cite{key-3} to exchange the secret key. Alice and Bob should agree on some integers:
$l,\, m,\, p,\, q,\, r,\, z$, where $z$ is $l$ bits long, $p+q=l+m,\,$ $p>m+q+r$, and $r>128$.
The protocol is then described as follows,

\begin{enumerate}
\item Alice and Bob agree upon the integers $l, m, p, q, r, z$. Alice randomly selects a private $m$ bit positive integer $x$, and Bob selects a private $m$ bit positive integer $y$.
\item Alice computes $U=\left\lfloor\frac{(xz)\mod2^p}{2^q}\right\rfloor$ and sends it to Bob.
\item Bob computes $V=\left\lfloor\frac{(yz)\mod2^p}{2^q}\right\rfloor$ and sends it to Alice.
\item Alice computes $W_a=\left\lfloor\frac{(xV)\mod2^{p-q}}{2^{r+m}}\right\rfloor$.
\item Bob computes $W_b=\left\lfloor\frac{(yU)\mod2^{p-q}}{2^{r+m}}\right\rfloor$.
\item The shared secret key is $K=W_a=W_b$ when $r>128$.
\end{enumerate}

\section{Practical attack to this cryptosystem }

Let $l,\, m,\, p,\, q,\, r$ and $z$ be fixed as above and let $F(x)=\left\lfloor\frac{(xz)\mod2^p}{2^q}\right\rfloor=u$, then we have \begin{equation*} 2^qu+y=xz \mod 2^p \end{equation*}for some integer $y$ with $0\leq y<2^q$, i.e.
\begin{equation}\label{eq:3} xz\equiv 2^qu+y \qquad(\modop 2^p).\end{equation}
Hence, finding an element  $x\in F^{-1}(\{u\})$ is equivalent to finding a proper vector $\left[\begin{array}{c}
x\\y\\
\end{array}\right]$ that satisfies equation ~\eqref{eq:3}, and $0\leq y<2^q$.

\begin{theorem}

All solutions to equation ~\eqref{eq:3} are of the form $\left[\begin{array}{c}
x\\y\\
\end{array}\right]=\left[\begin{array}{c}
x_0+\alpha_1x_1+\alpha_2x_2\\y_0+\alpha_1y_1+\alpha_2y_2\\
\end{array}\right]$, where $\alpha_1, \alpha_2\in \mathbb{Z}$ and \begin{equation*} x_0=\left\lceil\frac{2^qu}{z}\right\rceil,\qquad y_0=zx_0-2^qu;\end{equation*}
\begin{equation*}x_1=\left\lfloor\frac{2^qu}{z}\right\rfloor,\qquad y_1=zx_1-2^p;\end{equation*} \begin{equation*}x_2=\left\lfloor\frac{2^qu}{z}\right\rfloor+1,\qquad y_2=zx_2-2^p.\end{equation*}
\end{theorem}

\begin{proof}
Let $x_0, y_0, x_1, y_1$ and $x_2, y_2$ be the values that are defined above, then it is obvious that for all integers $\alpha_1$ and $\alpha_2$ we always have
\begin{equation*}(x_0+\alpha_1x_1+\alpha_2x_2)z\equiv 2^qu+(y_0+\alpha_1y_1+\alpha_2y_2)\qquad (\modop 2^p),\end{equation*} i.e. all vectors of the form $\left[\begin{array}{c}
x\\y\\
\end{array}\right]=\left[\begin{array}{c}
x_0+\alpha_1x_1+\alpha_2x_2\\y_0+\alpha_1y_1+\alpha_2y_2\\
\end{array}\right]$ are solutions to equation ~\eqref{eq:3} for $\forall \alpha_1, \alpha_2\in \mathbb{Z}$.

On the other hand, since \begin{equation*} x_0z\equiv 2^qu+y_0 \qquad(\modop 2^p),\end{equation*} let$\left[\begin{array}{c}
\hat{x}\\ \hat{y}\\
\end{array}\right]$ be an arbitrary solution to equation ~\eqref{eq:3}, then \begin{equation*} (\hat{x}-x_0)z\equiv \hat{y}-y_0 \qquad(\modop 2^p).\end{equation*}
Hence there exists $n\in \mathbb{Z}$, such that \begin{equation}\label{eq:4} (\hat{x}-x_0)z=n2^p+ \hat{y}-y_0 ,\end{equation} i.e.
 \begin{equation*}
\hat{x}-x_0=\frac{n2^p+ \hat{y}-y_0}{z}=n\left\lfloor \frac{2^p}{z}\right\rfloor+C+\frac{\hat{y}-y_0}{z}
\end{equation*} for some number $C\in\mathbb{R}$.

Since $\hat{x}-x_0$ is integer, $C+\frac{\hat{y}-y_0}{z}$ should also be an integer, call it $N$. Now we have
\begin{equation*}
\hat{x}-x_0=n\left\lfloor \frac{2^p}{z}\right\rfloor+N
=(n-N)\left\lfloor \frac{2^p}{z}\right\rfloor+N\left(\left\lfloor \frac{2^p}{z}\right\rfloor+1\right).
\end{equation*}
Let $n-N=\alpha_1$ and $N=\alpha_2$, we can get \begin{equation}\label{eq:5} \hat{x}-x_0=\alpha_1x_1+\alpha_2x_2.\end{equation}
Next, combining equations ~\eqref{eq:4} and ~\eqref{eq:5} we have
\begin{equation*}
\begin{split}
\hat{y}-y_0&=(\alpha_1x_1+\alpha_2x_2)z-n2^p\\
&=\alpha_1(zx_1-2^p)+\alpha_2(zx_2-2^p)\\
&=\alpha_1y_1+\alpha_2y_2.
\end{split}
\end{equation*}
So $\left[\begin{array}{c}
\hat{x}\\ \hat{y}\\
\end{array}\right]=\left[\begin{array}{c}
x_0+\alpha_1x_1+\alpha_2x_2\\y_0+\alpha_1y_1+\alpha_2y_2\\
\end{array}\right].$

\end{proof}

\begin{theorem}

The following algorithm can be used to find a minimal solution $ \left[\begin{array}{c}
x\\y\\
\end{array}\right]$ of equation ~\eqref{eq:3}, with respect to the norm induced by an arbitrary inner product $\langle -, -\rangle$ on $\mathbb{R}^2$.
\end{theorem}
\begin{algorithm}

\begin{enumerate}
\item Set $\mathbf{u_1}\leftarrow \left[\begin{array}{c}
x_1\\y_1\\
\end{array}\right], \mathbf{u_2}\leftarrow \left[\begin{array}{c}
x_2\\y_2\\
\end{array}\right]$ and done $\leftarrow 0$;
\item While $\text{done} =0$, do
\begin{itemize}
 \item$c_1\leftarrow Round \left(\frac{\langle \mathbf{u_1}, \mathbf{u_2}\rangle}{\langle \mathbf{u_2}, \mathbf{u_2}\rangle}\right)$;\qquad $\mathbf{u_1}\leftarrow \mathbf{u_1}-c_1\mathbf{u_2}$;
 \item$c_2\leftarrow Round \left(\frac{\langle \mathbf{u_1}, \mathbf{u_2}\rangle}{\langle \mathbf{u_1}, \mathbf{u_1}\rangle}\right)$;\qquad $\mathbf{u_2}\leftarrow \mathbf{u_2}-c_2\mathbf{u_1}$;
 \item if $c_1=0$ and $c_2=0$, then $\text{done}\leftarrow 1$.
 \end{itemize}
\item Solve the equation $[\mathbf{u_1}, \mathbf{u_2}]\left[\begin{array}{c}
\alpha_1\\\alpha_2\\
\end{array}\right]=\left[\begin{array}{c}
x_0\\y_0\\
\end{array}\right]$
\item $a_1\leftarrow\left\lfloor\alpha_1\right\rfloor$, $a_2\leftarrow\left\lfloor\alpha_2\right\rfloor$;
\item $\left[\begin{array}{c}
x\\y\\
\end{array}\right]\leftarrow\left[\begin{array}{c}
x_0\\y_0\\
\end{array}\right]-a_1\mathbf{u_1}-a_2\mathbf{u_2}$
\end{enumerate}
Note: In this algorithm, we let $Round\left(\pm\frac{1}{2}\right)=0$.
\end{algorithm}

\begin{proof}

First we show that the algorithm terminates.

Without loss of generality, we can assume $c_1= Round \left(\frac{\langle \mathbf{u_1}, \mathbf{u_2}\rangle}{\langle \mathbf{u_2}, \mathbf{u_2}\rangle}\right)\neq 0$, and let $\frac{\langle \mathbf{u_1}, \mathbf{u_2}\rangle}{\langle \mathbf{u_2}, \mathbf{u_2}\rangle}=c_1+\varepsilon$,  where $-\frac{1}{2}\leq \varepsilon\leq\frac{1}{2}$.
Then \begin{equation*}
\begin{split}\parallel \mathbf{u_1}-c_1\mathbf{u_2}\parallel^2&=\parallel \mathbf{u_1}\parallel^2+c_1^2\parallel \mathbf{u_2}\parallel^2-2c_1\langle \mathbf{u_1}, \mathbf{u_2}\rangle\\
&=\parallel \mathbf{u_1}\parallel^2-\parallel \mathbf{u_2}\parallel^2\left(2c_1\frac{\langle \mathbf{u_1}, \mathbf{u_2}\rangle}{\langle \mathbf{u_2},\mathbf{u_2}\rangle}-c_1^2\right)\\
&=\parallel \mathbf{u_1}\parallel^2-\parallel \mathbf{u_2}\parallel^2\left(2c_1\varepsilon+c_1^2\right)
\end{split}
\end{equation*}

Case I: Suppose $c_1>0$, then $c_1\geq 1$. Since $Round\left(\frac{1}{2}\right)=0$, we have either $-\frac{1}{2}<\varepsilon<0$ or $\varepsilon\geq0$. Hence $c_1+2\varepsilon>0$, and $2c_1\varepsilon+c_1^2\ >0$ as well.

Case II: Suppose $c_1<0$, then $c_1 \leq-1 $. Since$Round\left(-\frac{1}{2}\right)=0$, we have either $0<\varepsilon<\frac{1}{2}$ or $\varepsilon\leq0$. Now $c_1+2\varepsilon<0$, so $2c_1\varepsilon+c_1^2\ >0$.

So by both of these two cases, we always have $2c_1\varepsilon+c_1^2\ >0$, i.e. $\parallel \mathbf{u_1}-c_1\mathbf{u_2}\parallel^2<\parallel \mathbf{u_1}\parallel^2$. By a similar argument we can get that $\parallel \mathbf{u_2}-c_2\mathbf{u_1}\parallel^2<\parallel \mathbf{u_2}\parallel^2$. That means $\norm{\mathbf{u_1}}$ and $\norm{\mathbf{u_2}}$ are strictly decreasing. Since there are only finite number of elements in $\mathcal{L}$ with norm less than $\text{max}\left\{\norm{\left[\begin{array}{c}
x_1\\y_1\\
\end{array}\right]},\norm{\left[\begin{array}{c}
x_2\\y_2\\
\end{array}\right]}\right\}$, the algorithm must terminate.

Furthermore when $c_1=c_2=0$, it's trivial to see that \begin{equation*}\left|\langle \mathbf{u_1}, \mathbf{u_2}\rangle\right|\leq \frac{1}{2}\text{min}\{\left|\langle \mathbf{u_1}, \mathbf{u_1}\rangle\right|,\left|\langle \mathbf{u_2}, \mathbf{u_2}\rangle\right|\}.\end{equation*}

Now we show that $\{\mathbf{u_1}, \mathbf{u_2}\}$ is a basis of $\mathcal{L}$. By the algorithm, it's easy to see that $\mathbf{u_1}$ and $\mathbf{u_2}$ are linear combinations of $\left[\begin{array}{c}
x_1\\y_1\\
\end{array}\right]$ and $\left[\begin{array}{c}
x_2\\y_2\\
\end{array}\right]$, so $\text{Span}_\mathbb{Z}(\mathbf{u_1},\mathbf{u_2})\subseteq \mathcal{L}$; on the other hand, every step of the algorithm is invertible, we also have  $\mathcal{L}\subseteq\text{Span}_\mathbb{Z}(\mathbf{u_1},\mathbf{u_2})$. Hence after terminating, $\mathcal{L}=\text{Span}_\mathbb{Z}(\mathbf{u_1},\mathbf{u_2})$.

Since $\mathbf{u_1}$, $\mathbf{u_2}$ are linearly independent over $\mathbb{R}$, there exist $\alpha_1,\alpha_2$ in $\mathbb{R}$, such that \begin{eqnarray*}\left[\begin{array}{c}
x_0\\y_0\\
\end{array}\right]=\alpha_1\mathbf{u_1}+\alpha_2\mathbf{u_2}.\end{eqnarray*}
Let $a_1$, $a_2 \in \mathbb{N}$ with $a_1=Round(\alpha_1)$ and $a_2=Round(\alpha_2)$, now we want to show that $\norm{\left[\begin{array}{c}
x_0\\y_0\\
\end{array}\right]-a_1\mathbf{u_1}-a_2\mathbf{u_2}}$ is minimized.
Let $z=b_1\mathbf{u_1}+b_2\mathbf{u_2}$ be an arbitrary vector in $\mathcal{L}$, and let\begin{equation}\label{eq:2} d=\norm{\left[\begin{array}{c}
x_0\\y_0\\
\end{array}\right]-z}^2 -\norm{\left[\begin{array}{c}
x_0\\y_0\\
\end{array}\right]-a_1\mathbf{u_1}-a_2\mathbf{u_2}}^2,\end{equation}we have the following cases:

Case 1: $a_1=b_1, a_2=b_2$, then $d=0$.

Case 2: $a_1\neq b_1, a_2=b_2$, then
\begin{eqnarray*}
d
&=&\norm{\left(\alpha_1-b_1\right)\mathbf{u_1}+\left(\alpha_2-b_2\right)\mathbf{u_2}}^2-\norm{\left(\alpha_1-a_1\right)\mathbf{u_1}+\left(\alpha_2-a_2\right)\mathbf{u_2}}^2\\
&=&(\alpha_1-b_1)^2\norm{\mathbf{u_1}}^2+(\alpha_2-b_2)^2\norm{\mathbf{u_2}}^2+2(\alpha_1-b_1)(\alpha_2-b_2)\langle \mathbf{u_1}, \mathbf{u_2}\rangle\\
&&-(\alpha_1-a_1)^2\norm{\mathbf{u_1}}^2-(\alpha_2-a_2)^2\norm{\mathbf{u_2}}^2-2(\alpha_1-a_1)(\alpha_2-a_2)\langle \mathbf{u_1}, \mathbf{u_2}\rangle\\
&=&\left(\alpha_1-b_1+\alpha_1-a_1\right)\left(a_1-b_1\right)\norm{\mathbf{u_1}}^2+2\left(a_1-b_1\right)\left(\alpha_2-a_2\right)\langle \mathbf{u_1}, \mathbf{u_2}\rangle\\
&\geq&2\left(\alpha_1-b_1+\alpha_1-a_1\right)\left(a_1-b_1\right)\left|\langle \mathbf{u_1}, \mathbf{u_2}\rangle\right|+2\left(a_1-b_1\right)\left(\alpha_2-a_2\right)\langle \mathbf{u_1}, \mathbf{u_2}\rangle.
\end{eqnarray*}If $a_1>b_1$, we have $ a_1-b_1>0,\quad 2\alpha_1-a_1-b_1\geq 1\geq \mid\alpha_2-a_2\mid,$ i.e. $d\geq0$;
if $a_1<b_1$, then $ a_1-b_1<0,\quad 2\alpha_1-a_1-b_1\leq -1$, but $\mid2\alpha_1-a_1-b_1\mid\geq \mid\alpha_2-a_2\mid,$ we still have $d\geq0$.

Case 3: $a_1= b_1, a_2\neq b_2$, this is the same as Case 2.

Case 4: $a_1\neq b_1, a_2\neq b_2$, then
\begin{eqnarray*}
d
&=&\norm{\left(\alpha_1-b_1\right)\mathbf{u_1}+\left(\alpha_2-b_2\right)\mathbf{u_2}}^2-\norm{\left(\alpha_1-b_1\right)\mathbf{u_1}+\left(\alpha_2-a_2\right)\mathbf{u_2}}^2\\
&&+\norm{\left(\alpha_1-b_1\right)\mathbf{u_1}+\left(\alpha_2-a_2\right)\mathbf{u_2}}^2-\norm{\left(\alpha_1-a_1\right)\mathbf{u_1}+\left(\alpha_2-a_2\right)\mathbf{u_2}}^2\\
&\geq&0.
\end{eqnarray*}
The inequality is because of Case 2 and 3.

Above all, the norm of the vector $\left[\begin{array}{c}
x_0\\y_0\\
\end{array}\right]-a_1\mathbf{u_1}-a_2\mathbf{u_2}$ is minimized.

 Since the vectors $\left[\begin{array}{c}
x_1\\y_1\\
\end{array}\right]$ and $\left[\begin{array}{c}
x_2\\y_2\\
\end{array}\right]$ are two solutions to the equation $zx\equiv y (\modop 2^p)$, all linear combinations of these two vectors are also solutions to this equation, in particular, $a_1\mathbf{u_1}+a_2\mathbf{u_2}$ is a solution; on the other hand $\left[\begin{array}{c}
x_0\\y_0\\
\end{array}\right]$ satisfies the equation $xz\equiv 2^qu+y (\modop 2^p)$. Hence the vector $\left[\begin{array}{c}
x_0\\y_0\\
\end{array}\right]-a_1\mathbf{u_1}-a_2\mathbf{u_2}$ is a solution to the equation $xz\equiv 2^qu+y (\modop 2^p)$, and it is minimal by previous result.

So, by all of the above arguments we can see that the algorithm can be used to find the minimal solution of equation ~\eqref{eq:3}.
\end{proof}

But only find the minimal solution of equation ~\eqref{eq:3} is still not enough to break the cryptosystem, because according to the system the solution should also satisfy $$\begin{cases} 0\leq x<2^m\\0\leq2^qu+y<2^m\\0\leq y<2^q ,\end{cases}$$ or, equivalently $$\begin{cases} 0\leq x<2^m=:B_1\\0\leq y<\text{min}\{2^q, 2^m-2^qu\}=:B_2.\end{cases}$$
To fix our algorithm, we define an inner product on $ \mathbb{R}^2$ by \begin{equation*}
\left\langle\left[\begin{array}{c}
a_1\\b_1\\
\end{array}\right],\left[\begin{array}{c}
a_2\\b_2\\
\end{array}\right]\right\rangle=a_1a_2+\left(\frac{B_1}{B_2}\right)^2b_1b_2.
\end{equation*}
It is easy to see that even with this new inner product, the proof in theorem 2 is still true.
Let $\{\mathbf{u_1}, \mathbf{u_2}\}$ be the minimal basis that we have found in the algorithm, and also let $\mathbf{v}=(x_0, y_0)^T$. Then we can write all four corners of the rectangle which is bounded by $\mathbf{v}$, $\mathbf{v}-(B_1, 0)^T$, $\mathbf{v}-(0, B_2)^T$ and $\mathbf{v}-(B_1, B_2)^T$ as real linear combinations of $\mathbf{u_1}$, $\mathbf{u_2}$, and then use this to find a lattice point within the bounding region (See the figure). By assumption, we know that there is at least one such lattice point; there could be more than one, but any one will solve the problem at hand.
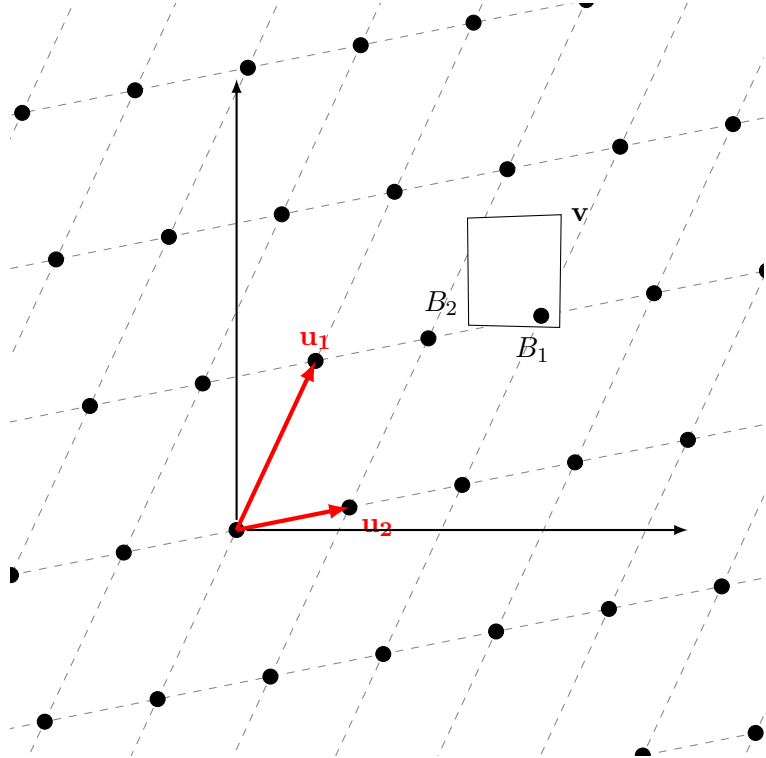
\begin{figure}[h!]
  \centering
\begin{tikzpicture}[>=latex]
    \begin{scope}
    \clip (0,0) rectangle (10cm,10cm); 
    \pgftransformcm{1}{0.2}{0.7}{1.5}{\pgfpoint{3cm}{3cm}} 

    \draw[style=help lines,dashed] (-14,-14) grid[step=1.5cm] (14,14); 
\filldraw[fill=white, draw=black] (2,1.55) node [above left ] {$B_2$} --  (1.25,2.6) -- (2.6,2.45) node [black, right ] {$\mathbf{v}$} -- (3.35,1.35) -- cycle  node [ below left] {$B_1$};
    \foreach \x in {-7,-6,...,7}{                           
        \foreach \y in {-7,-6,...,7}{                       
        \node[draw,circle,inner sep=2pt,fill] at (1.5*\x,1.5*\y) {}; 
        }
    }
    \draw[ultra thick,red,->] (0,0) -- (0,1.5)  node [above ] {$\mathbf{u_1}$};
    \draw[ultra thick,red,->] (0,0) -- (1.5,0) node [below right] {$\mathbf{u_2}$};
    \node (O) at (0,0) {};
    \end{scope}
\draw[->,thick] (O) -- ++(0,6);
\draw[->,thick] (O) -- ++(6,0);
\end{tikzpicture}

\end{figure}

An example based on our algorithm is as follows,

\begin{example}
Alice chooses her secret key $X=12345$, which is 14-bit long, Bob and she agree on some common integers
$Z=6173$, $q=5$ and $p=22$, then by Bouftass's protocol, Alice needs to send the number
 \begin{equation*}
 U=\left\lfloor\frac{(XZ)\mathrm{mod}(2^p)}{2^q}\right\rfloor=708192
 \end{equation*}to Bob.

 To recover Alice's secret key $X$, Eve can use the above algorithm to get:
 \begin{equation*}
 x_0=115,\qquad y_0=1703,\qquad \mathbf{u_1}=\left[\begin{array}{c}
-25140\\28\\
\end{array}\right] \qquad\text{and}\qquad \mathbf{u_2}=\left[\begin{array}{c}
-33973\\-129\\
\end{array}\right].
\end{equation*}
By computing the four corners, $\mathbf{v}=13.790\mathbf{u_1}-10.208\mathbf{u_2}$, $\mathbf{v}-(B_1, 0)^T=14.252\mathbf{u_1}-10.108\mathbf{u_2}$, $\mathbf{v}-(0, B_2)^T=13.531\mathbf{u_1}-10.016\mathbf{u_2}$ and $\mathbf{v}-(B_1, B_2)^T=13.992\mathbf{u_1}-9.916\mathbf{u_2}$, Eve will find that \begin{equation*}\left[\begin{array}{c}
x\\y\\
\end{array}\right]=\left[\begin{array}{c}
115\\1703\\
\end{array}\right]-14\left[\begin{array}{c}
-25140\\28\\
\end{array}\right]+10\left[\begin{array}{c}
-33973\\-129\\
\end{array}\right]=\left[\begin{array}{c}
12345\\21\\
\end{array}\right],\end{equation*} i.e. $x=12345$ and $y=21$.

\end{example}

\section{Conclusion}

In this paper, we provide a practical attack to Bouftass's cryptosystem based on Gaussian lattice reduction. Our attack is simple and fast, it works when the conditions $l+m=p+q$ and $p>m+q$ are satisfied. We proved that our algorithm can definitely find a solution to the equation $u=\left\lfloor\frac{(xz)\mathrm{mod}(2^p)}{2^q}\right\rfloor$ , but the solution is not necessarily unique.

We also remark that a similar approach using LLL algorithm seems to work in practice, but the method presented here admitted an easier proof.

\end{document}